\documentclass[11pt,a4paper]{article}
\usepackage[T1]{fontenc}
\usepackage{lmodern,microtype}
\usepackage[margin=25mm]{geometry}
\usepackage{amsmath,amssymb,amsthm,mathtools,bm}
\usepackage{enumitem,booktabs}
\usepackage{xcolor}
\usepackage[colorlinks=true,linkcolor=blue!45!black,citecolor=blue!45!black,urlcolor=blue!45!black]{hyperref}
\hypersetup{pdftitle={Limits of adaptive two-way linear postprocessing in six-state quantum key distribution},pdfauthor={Toyohiro Tsurumaru}}
\setlist{itemsep=2pt,topsep=4pt}
\newtheorem{theorem}{Theorem}
\newtheorem{lemma}[theorem]{Lemma}
\newcommand{\E}{\mathbb E}
\newcommand{\Prob}{\mathbb P}
\DeclareMathOperator{\rank}{rank}
\DeclareMathOperator{\Bern}{Bern}
\newcommand{\F}{\mathbb F_2}
\newcommand{\pg}{p_{\mathrm{guess}}}
\newcommand{\wt}{\operatorname{wt}}
\newcommand{\Tr}{\operatorname{Tr}}
\newcommand{\ket}[1]{\lvert #1\rangle}
\newcommand{\bra}[1]{\langle #1\rvert}
\newcommand{\proj}[1]{\ket{#1}\!\bra{#1}}
\newcommand{\acc}{\mathrm{acc}}
\newcommand{\eps}{\varepsilon}
\title{Limits of adaptive two-way linear postprocessing\\in six-state quantum key distribution}
\author{Toyohiro Tsurumaru\\[3pt]
\small Mitsubishi Electric Corporation, Information Technology R\&D Center\\
\small 5-1-1 Ofuna, Kamakura-shi, Kanagawa 247-8501, Japan}
\date{}
\begin{document}
\maketitle
\begin{abstract}
Two-way classical communication improves the noise tolerance of quantum key distribution, but the limits of adaptive parity processing remain unclear. For the six-state protocol with equal three-basis error rate $p\in[p_*,1/2]$, where $p_*=(5-\sqrt5)/10\simeq27.64\%$, we construct a single compatible collective attack that defeats a broad class of linear postprocessing. Each public message may be a linear function of the sender's raw key, chosen adaptively from the public history. For every finite input length and every accepted transcript of positive probability, Eve's optimal probability of recovering any linear function of Alice's raw key is at least Bob's. Consequently, reliable agreement on a linear final key is incompatible with secrecy. The same conclusion holds for arbitrary publicly seeded privacy amplification after reliable recovery of a linear candidate. The proof conditions on the complete transcript, which restricts the two raw keys to affine subspaces. A local degradation of Eve's state, a known R\'enyi-entropy comparison, and a coset-probability inequality then give the recovery comparison. Both inequalities are proved in self-contained appendices. The result allows block splitting, regrouping, nonzero-syndrome selection, and adaptive linear reconciliation.
\end{abstract}

\section{Introduction}
Two-way classical communication can improve the noise tolerance of quantum key distribution (QKD). For the six-state protocol, bit-error detection and parity processing yield secure keys for QBERs strictly below the threshold~\cite{GottesmanLo,Chau,Bae}
\begin{equation}
 p_*:=\frac{5-\sqrt5}{10}\simeq0.2763932023. \label{eq:threshold}
\end{equation}
In contrast, an equal three-basis quantum bit error rate (QBER) of $1/3$ is compatible with an intercept--resend attack, which excludes secret-key generation by any classical postprocessing. The interval between these values therefore poses a natural question: can more general classical processing improve on the established threshold?~\cite{Khatri}

Linear codes provide a systematic way to enlarge parity-based protocols. They permit arbitrary checks, retention of multiple bits, and selection among different error syndromes. Du \emph{et al.}\ recently developed an advantage-distillation framework based on general linear codes and conjectured that repetition codes provide the best threshold within their framework~\cite{Du}. They also identified adaptive coding over multiple rounds as a direction for further study. A converse for fixed repetition-code constructions does not by itself settle what such more general processing can accomplish.

We prove a finite-length recovery converse for adaptive linear public communication. At each round, either party may announce a linear function of its original raw key. The sender and matrix may depend on the entire preceding public history, including public randomness. At the end, Alice's candidate is a linear function of her raw key, again chosen from that history. Bob may retain his complete raw key and use any decoder. Thus, the restriction is on the disclosed functions and Alice's output, not on the complexity or linearity of Bob's decoding computation.

For every $p\in[p_*,1/2]$, one collective attack, independent of the protocol, makes Eve at least as successful as Bob at recovering Alice's candidate. The comparison holds separately for every accepted public transcript of positive probability. It therefore survives arbitrary selection based on that transcript and does not require acceptance with probability approaching one. It also rules out direct agreement on a linear final key, without assuming recovery of a longer intermediate string.

The proof uses the affine structure of a fixed communication transcript. Alice's and Bob's compatible inputs form a product of two affine subspaces. Eve knows the raw error pattern under our attack; conditioning additionally on that pattern leaves Alice uniform on an affine subspace whose direction is the intersection of the two linear spaces. We evaluate Eve's information on that subspace and compare it with Bob's posterior coset probabilities. This avoids tracking surviving blocks through successive processing stages.

The conclusion is restricted to the stated class. In particular, it does not establish zero unrestricted secret-key capacity. Its relation to existing advantage-distillation and symmetric-extension results is discussed in Sec.~\ref{sec:scope}.

\section{Main result}\label{sec:model}
\subsection{Protocol class}
We now specify the six-state QKD protocols and the class of classical postprocessing considered in this paper.

Alice prepares uniformly chosen eigenstates of the Pauli $X,Y,Z$ observables, and Bob measures in independently chosen bases. We consider equal observed error rates $p_X=p_Y=p_Z=p$. After sifting, the undisclosed matched-$Z$ data form the raw keys. A converse needs only one collective attack compatible with these three error rates. Under such an attack, separate parameter-estimation samples can be conditioned on without changing the remaining independent raw pairs.

Fix a finite raw length $n$ and write the raw keys as column vectors $A,B\in\F^n$. Public communication is authenticated and visible to Eve. Public randomness, if used, consists of finite random seeds drawn independently of the initial joint state of the raw keys and Eve at the start of classical postprocessing. The public transcript records these seeds as they are announced and all messages of the form below; the public history $h$ is the portion disclosed so far. Conditioned on the seeds, the protocol is a finite communication tree. At a node with public history $h$, the protocol chooses its sender and announces either
\begin{equation}
 M_hA \quad\text{or}\quad M_hB. \label{eq:message}
\end{equation}
The sender, binary matrix $M_h$, and message length are determined by $h$. Subsequent choices, termination, and acceptance also depend only on the public history. A publicly determined control flag adds no constraint beyond the messages already in that history.

At an accepted leaf with complete transcript $T=t$, Alice outputs
\begin{equation}
 K=F_tA\in\F^{\ell(t)}. \label{eq:output}
\end{equation}
The binary matrix $F_t$ and output length $\ell(t)$ are determined by $t$; $F_t$ need not have full rank. Bob estimates $K$ from $(B,t)$; his decoder may be nonlinear or randomized. There is no bound on computational complexity. Alice's output and the public messages do not use additional secret local randomness. Giving Bob his entire original word allows him to exploit even coordinates omitted from the candidate.

This definition permits parity tests on arbitrary coordinate sets, adaptive block splitting and regrouping, coordinate permutations, and acceptance of any publicly specified syndrome classes. For example, after publishing $HA$ and $HB$, the parties may retain a block for any chosen values of $H(A+B)$, then select their next checks from those values. All subsequent linear reconciliation messages may be included in the same transcript.

Public-history-dependent offsets could also be added to Eqs.~\eqref{eq:message} and~\eqref{eq:output}. They amount to reversible relabelings of the public messages and the output key, determined by the public history. These relabelings preserve recovery probabilities and secrecy, so we omit the offsets without loss of generality.

Special cases include the standard B/P processing of Refs.~\cite{GottesmanLo,Chau} and the public-parity-check formulation of repetition-code advantage distillation~\cite{Bae}, with retained bits given by linear functions.

\subsection{Recovery comparison}
For a classical variable $X$ and side information $Q$, let $\pg(X\mid Q)$ be the optimal probability of guessing the entire value of $X$, optimized over measurements when $Q$ is quantum. All conditional states below are normalized.

\begin{theorem}[Recovery under adaptive linear communication]\label{thm:main}
For each $p\in[p_*,1/2]$, there is a collective attack $\mathcal A_p$, compatible with QBER $p$ in all three bases, such that the following holds for every finite $n$ and every protocol defined above. For every accepted transcript $t$ of positive probability,
\begin{equation}
 \pg(K\mid E,T=t)\ge \pg(K\mid B,T=t). \label{eq:main}
\end{equation}
Here $E$ contains Eve's systems for all $n$ raw positions. The attack depends only on $p$, not on the protocol, its public seeds, or the output map.
\end{theorem}

The statement is a comparison with Bob's optimal decoder, so it also applies to any decoder used in an implementation. Eve's measurement may depend on $t$ and may be collective; no efficient implementation of it is claimed.

Suppose accepted transcripts produce a fixed $\ell\ge1$ bit key and Bob's error probability conditioned on acceptance is $\eps_B$. Averaging Eq.~\eqref{eq:main} gives Eve a guessing probability of at least $1-\eps_B$. Consequently, the conditional secrecy distance obeys
\begin{equation}
 \Delta_{\acc}:=\tfrac12\bigl\|\rho_{KET\mid\acc}
 -\tau_\ell\otimes\rho_{ET\mid\acc}\bigr\|_1
 \ge 1-\eps_B-2^{-\ell}, \label{eq:security}
\end{equation}
where $\tau_\ell$ is the uniform key state. Indeed, a measurement that guesses $K$ succeeds with probability $2^{-\ell}$ on the ideal state. Reliable agreement therefore precludes conditional secrecy. If acceptance has probability $a$, the usual acceptance-weighted secrecy distance is $a\Delta_{\acc}$; Eq.~\eqref{eq:security} does not assert a nonvanishing unconditional distance when $a\to0$. Variable output lengths can be handled transcript by transcript.

The theorem also applies when $K$ is a longer \emph{candidate} that Bob reconstructs before privacy amplification. Eve can reconstruct that candidate with no larger error and apply the same public function. Thus arbitrary publicly seeded privacy amplification, including nonlinear functions, cannot produce a secret nonempty key from a reliably recovered candidate. This conclusion assumes that any additional communication or acceptance test is covered by the stated protocol class.

As a consequence, the Devetak--Winter key-rate lower bound~\cite[Theorem~2.1]{DW} is nonpositive under $\mathcal A_p$ for an independently repeatable finite-block output with a fixed-length linear candidate in this protocol class. Indeed, a positive value would enable linear information reconciliation and publicly seeded linear privacy amplification~\cite{SW,RK}, yielding a secure protocol within the same class and contradicting Theorem~\ref{thm:main}.

\subsection{Scope and relation to previous work}\label{sec:scope}
Theorem~\ref{thm:main} includes arbitrary linear parity checks and output maps, with multiple retained bits. It also permits choices based on previously disclosed parity values. Splitting a retained block, combining it with another block, or returning to previously unused coordinates changes the matrices in Eq.~\eqref{eq:message}, but not the fixed-transcript argument.

For the framework of Du \emph{et al.}~\cite{Du}, the theorem covers the processing based on publicly disclosed linear syndromes without one-time-pad encryption or secret noise addition, including retention of nonzero error-syndrome classes. It also covers subsequent rounds whose checks depend on the disclosed syndromes and linear reconciliation of the retained data. Their full framework includes further resources and preprocessing choices, so the present result does not resolve every version of their threshold conjecture. The precise conclusion is that the covered linear interactions and candidate-recovery procedures cannot surpass $p_*$.

Earlier analyses establish the same threshold for repetition-code advantage distillation followed by arbitrary one-way key distillation~\cite{Bae}. Symmetric-extension arguments exclude arbitrary subsequent one-way key distillation for certain selection classes: in particular, linear parity selection retaining one logical bit~\cite{Myhr2009}, and linear advantage distillation retaining two logical bits at arbitrary block length~\cite[Chap.~9]{MyhrThesis}. Khatri and L\"utkenhaus studied a broader selection problem and supplied analytical and numerical evidence for a general impossibility conjecture~\cite{Khatri}. Our recovery statement accommodates an arbitrary number of bits and adaptive linear interaction, but does not establish preservation of symmetric extension or exclude arbitrary nonlinear distillation. These are different boundaries on the protocol class.

A secretly randomized or nonlinear Alice output is not covered by Eq.~\eqref{eq:output}. Nonlinear privacy amplification is covered only when Bob reliably recovers a linear candidate first. Secret noisy preprocessing and directly shared nonlinear outputs therefore remain outside this result.

\section{Proof of Theorem~\ref{thm:main}}\label{sec:proof}
All vector spaces and matrix operations are over $\F$, and all logarithms are base two. For $0<q\le1$, a subspace $C\subseteq\F^n$, and a coset $w+C$, define
\begin{equation}
 Z_{w+C}(q):=\sum_{c\in C}q^{\wt(w+c)},\qquad
 Z_C(q):=Z_{0+C}(q). \label{eq:weights}
\end{equation}
All weight sums in the transcript proof are evaluated at the same positive parameter $\alpha$.

\subsection{A common attack with a pure-state alphabet}\label{sec:attack}
For $p\in[p_*,1/2]$, put
\begin{equation}
 \alpha=\frac{p}{1-p},\qquad \gamma=\sqrt\alpha,\qquad
 \ket{\phi_a}=\sqrt{\frac{1+\gamma}{2}}\ket0
             +(-1)^a\sqrt{\frac{1-\gamma}{2}}\ket1.
 \label{eq:alphabet}
\end{equation}
Our attack produces, on each matched-$Z$ pair, the state
\begin{equation}
 \omega_p=\frac12\sum_{a,e\in\F}p_e\,
 \proj a_A\otimes\proj{a+e}_B\otimes\proj e_{E_1}
 \otimes\proj{\phi_a}_{E_2},\qquad p_0=1-p,\quad p_1=p.
 \label{eq:attack}
\end{equation}
Thus $A$ is uniform, $B=A+N$, and $N$ has independent Bernoulli$(p)$ coordinates. Eve holds $N$ exactly in $E_1$ and the product state
\begin{equation}
 \ket{\phi_a}=\bigotimes_{j=1}^n\ket{\phi_{a_j}} \label{eq:product}
\end{equation}
in $E_2$. We use the same notation for the single-bit and product states, with their arguments specifying the length.

\begin{lemma}[Physical realization]\label{lem:attack}
The state~\eqref{eq:attack} is obtainable by a basis-independent collective attack compatible with equal three-basis QBER $p$, followed by a local channel on Eve's system.
\end{lemma}

The proof of Lemma~\ref{lem:attack}, including the attack construction, is given in Sec.~\ref{sec:attack-realization}. Before the local degradation, the overlap of Eve's two states in the no-bit-error branch is
\begin{equation}
 \beta=\frac{1-2p}{1-p}.
\end{equation}
She can map these states to the alphabet in Eq.~\eqref{eq:alphabet} precisely in the required range
\begin{equation}
 \beta\le\gamma
 \quad\Longleftrightarrow\quad
 (1-2p)^2\le p(1-p)
 \quad\Longleftrightarrow\quad p\ge p_* \qquad(0\le p\le1/2).
 \label{eq:overlap}
\end{equation}
The bit-error branch consists of orthogonal states and can be degraded to the same alphabet. Choosing the larger overlap $\gamma=\sqrt\alpha$ deliberately weakens Eve so that the same parameter $\alpha$ controls both her squared state overlaps and Bob's posterior error probabilities. The original overlap $\beta$ is needed only to justify this physical realization.

\subsection{Eve's information on an affine subspace}\label{sec:entropy}
We first identify the spectrum of Eve's average state with a classical probability distribution. Write $S(\rho)=-\Tr\rho\log\rho$ and $H_2(\rho)=-\log\Tr\rho^2$. For a classical variable, $H$ and $H_2$ denote its Shannon and collision entropies; $H(X\mid Q)$ denotes conditional von Neumann entropy when $Q$ is quantum. For a subspace $C\subseteq\F^n$, define
\begin{equation}
 \rho_C=\frac1{|C|}\sum_{a\in C}\proj{\phi_a},
 \label{eq:average-state}
\end{equation}
using the product states in Eq.~\eqref{eq:product}.

\begin{lemma}[Classical representation of the spectrum]\label{lem:spectrum}
Let $C\subseteq\F^n$ have dimension $k$, and let $G\in\F^{k\times n}$ be a full-row-rank generator matrix of $C$. For $0\le\gamma\le1$, let $W=(W_1,\ldots,W_n)^{\mathsf T}$ have independent Bernoulli$((1-\gamma)/2)$ coordinates. The eigenvalues of $\rho_C$ are the list
\begin{equation}
 \lambda_C(s)=\Pr[GW=s],\qquad s\in\F^k,
 \label{eq:spectrum-law}
\end{equation}
supplemented with $2^n-2^k$ zeros. In particular,
\begin{equation}
 S(\rho_C)=H(GW),\qquad H_2(\rho_C)=H_2(GW).
 \label{eq:spectrum-entropies}
\end{equation}
\end{lemma}

The proof is given in Sec.~\ref{sec:spectrum-proof}. This is a spectral identity, not a measurement that turns Eve's system into a classical variable. For the attack, we use $\gamma=\sqrt\alpha>0$ and obtain the following entropy bound.

\begin{lemma}[Entropy of a linear image]\label{lem:entropy}
Let $U$ be uniform on an affine space $u+D\subseteq\F^n$, let $X=FU$, and let $Q$ contain $\ket{\phi_U}$ from Eq.~\eqref{eq:product}. Put $D_0=D\cap\ker F$. Then
\begin{equation}
 H(X\mid Q)\le \log\frac{Z_D(\gamma^2)}{Z_{D_0}(\gamma^2)}.
 \label{eq:entropy-bound}
\end{equation}
\end{lemma}

Here $D$ is the direction subspace of the allowed affine set $u+D$. Two allowed values give the same $X$ exactly when their difference lies in $D_0$; thus $X$ labels an element of $D/D_0$, up to a known relabeling.

\begin{proof}
Translation by $u$ is accompanied by the unitary $Z^u$ on $Q$, because $Z^u\ket{\phi_a}=\ket{\phi_{a+u}}$. It also relabels $X$ by $Fu$, so it suffices to take $u=0$. Let $r=\dim D-\dim D_0$. The $2^r$ output values are equiprobable. Their conditional states are unitary translates of $\rho_{D_0}$, giving
\begin{equation}
 H(X\mid Q)=r+S(\rho_{D_0})-S(\rho_D). \label{eq:quotient-entropy}
\end{equation}

Choose a generator matrix $G$ of $D$ whose last rows form a generator $G_0$ of $D_0$. By Lemma~\ref{lem:spectrum}, the spectra of $\rho_D$ and $\rho_{D_0}$ are represented by $GW$ and $G_0W$ for the same independent Bernoulli vector $W$. The latter is a coordinate projection of the former. The binary-linear R\'enyi comparison~\cite[Theorem~1]{Vijay}, proved directly in Appendix~\ref{app:detailed-proof}, yields
\begin{equation}
 S(\rho_D)-S(\rho_{D_0})
 \ge H_2(\rho_D)-H_2(\rho_{D_0}). \label{eq:renyi}
\end{equation}
The required comparison is $H(L_1W)-H(L_2W)\ge H_2(L_1W)-H_2(L_2W)$ for independent binary coordinates and nested linear images, $\ker L_1\subseteq\ker L_2$.

Since $\langle\phi_a|\phi_b\rangle=\gamma^{\wt(a+b)}$, each difference in $C$ occurs $|C|$ times in the purity calculation. Hence
\begin{equation}
 \Tr\rho_C^2=\frac{Z_C(\gamma^2)}{|C|},\qquad
 H_2(\rho_C)=\dim C-\log Z_C(\gamma^2). \label{eq:purity}
\end{equation}
Substituting Eqs.~\eqref{eq:renyi} and~\eqref{eq:purity} into Eq.~\eqref{eq:quotient-entropy} proves the bound.
\end{proof}

For classical--quantum states, $\pg(X\mid Q)=2^{-H_{\min}(X\mid Q)}$ and $H_{\min}(X\mid Q)\le H(X\mid Q)$~\cite{KRS}. Lemma~\ref{lem:entropy} therefore gives
\begin{equation}
 \pg(X\mid Q)\ge\frac{Z_{D_0}(\gamma^2)}{Z_D(\gamma^2)}. \label{eq:eve-bound}
\end{equation}

\subsection{Two comparisons of code weight sums}\label{sec:cosets}
\begin{lemma}[Coset and subspace comparisons]\label{lem:cosets}
Fix $0<q\le1$.
If $C_0\subseteq C$, then, for every $w\in\F^n$,
\begin{equation}
 \frac{Z_{w+C_0}(q)}{Z_{w+C}(q)}
 \le\frac{Z_{C_0}(q)}{Z_C(q)}. \label{eq:coset}
\end{equation}
If $D,H\subseteq C$ and $J=D\cap H$, then
\begin{equation}
 \frac{Z_J(q)}{Z_D(q)}\ge\frac{Z_H(q)}{Z_C(q)}. \label{eq:subspace}
\end{equation}
\end{lemma}

\begin{proof}
We use the following coset inequality of Cammarota and Russo~\cite[Proposition~3.1, Eq.~(3.1)]{CR}:
\begin{equation}
 Z_{u+G}(q)Z_{v+G}(q)\le Z_G(q)Z_{u+v+G}(q) \label{eq:cr}
\end{equation}
for any subspace $G$ and vectors $u,v$. In their notation, the Bernoulli product measure $\mu_x$ with common parameter $x=q/(1+q)\le1/2$ satisfies $\mu_x(w+G)=(1+q)^{-n}Z_{w+G}(q)$, so the normalization cancels from the two sides. Appendix~\ref{app:coset-proof} gives an elementary counting proof.

Set $G=C_0$, $u=w$, and sum over representatives $v$ of $C/C_0$. This gives $Z_{w+C_0}Z_C\le Z_{C_0}Z_{w+C}$, proving Eq.~\eqref{eq:coset}.

Next set $G=J$ and sum Eq.~\eqref{eq:cr} over representatives $u$ of $D/J$ and $v$ of $H/J$. Since $D\cap H=J$, the sums $u+v+J$ run once through $(D+H)/J$. Thus
\begin{equation}
 Z_DZ_H\le Z_JZ_{D+H}\le Z_JZ_C,
\end{equation}
which proves Eq.~\eqref{eq:subspace}.
\end{proof}

\subsection{Fixing the transcript}\label{sec:transcript}
Even when the choices are adaptive, fixing the complete transcript collects all disclosed syndromes into linear constraints on the original raw keys. This avoids tracking the distribution of surviving data through successive rounds.

\begin{proof}[Proof of Theorem~\ref{thm:main}]
Fix the public seeds and a positive-probability accepted transcript $t$. Stack all equations announced by Alice along this path into $M_AA=s_A$. Similarly write Bob's equations as $M_BB=s_B$. The set of input pairs generating this path is exactly
\begin{equation}
 A\in a_t+C_A,\qquad B\in b_t+C_B,\qquad
 C_A=\ker M_A,\quad C_B=\ker M_B. \label{eq:rectangle}
\end{equation}
Here $a_t,b_t$ are any solutions of the respective syndrome equations. Indeed, once the preceding messages agree with $t$, the next sender and linear map are fixed. Induction through the path proves both necessity and sufficiency of these equations. Public stopping and acceptance rules introduce no further private-input restrictions.

For $K=F_tA$, define
\begin{equation}
 D=C_A\cap C_B,\qquad H=C_A\cap\ker F_t,\qquad
 J=D\cap\ker F_t=D\cap H. \label{eq:spaces}
\end{equation}

\emph{Eve's recovery.} Eve has the error pattern $N$ in her classical register $E_1$. Condition additionally on $N=e$. She converts Bob's equations into $M_BA=s_B+M_Be$, obtaining further syndromes of Alice's key. Because $A$ is initially uniform and independent of $N$, it is now uniform on
\begin{equation}
 (a_t+C_A)\cap(b_t+e+C_B). \label{eq:intersection}
\end{equation}
Whenever this intersection is nonempty, it is a coset of $D$. Its offset is determined by $(e,t)$ and can be removed by a known relabeling and a unitary on $E_2$. Lemma~\ref{lem:entropy}, with $D_0=J$ and $\gamma^2=\alpha$, gives
\begin{equation}
 H(K\mid E_2,N=e,T=t)\le\log\frac{Z_D(\alpha)}{Z_J(\alpha)}.
\end{equation}
The bound is independent of $e$. Averaging over the flag, then using $H_{\min}\le H$, yields
\begin{equation}
 \pg(K\mid E,T=t)\ge\frac{Z_J(\alpha)}{Z_D(\alpha)}. \label{eq:eve-transcript}
\end{equation}
This argument does not assume that the error flag remains independent of the candidate after an interactive transcript.

\emph{Bob's recovery.} Fix a compatible value $B=b$. At fixed $(b,t)$, Bob's own messages add no constraint on $A$, so its allowed set is $a_t+C_A$. Initially, $A$ is uniform and independent of $N$; hence $N=b+A$ has conditional distribution
\begin{equation}
 \Pr[N=z\mid B=b,T=t]
 =\frac{\alpha^{\wt(z)}}{Z_{b+a_t+C_A}(\alpha)}
 \quad(z\in b+a_t+C_A). \label{eq:bob-posterior}
\end{equation}
Since $K=F_t(b+N)$, two allowed errors give the same key exactly when their difference lies in $H=C_A\cap\ker F_t$. Bob thus needs only the coset modulo $H$, whereas Eve's corresponding subspace is $J$. For each coset $w+H$ in this affine space, $w+C_A=b+a_t+C_A$. Summing Eq.~\eqref{eq:bob-posterior} over $w+H$ and applying Eq.~\eqref{eq:coset} gives
\begin{equation}
 \Pr[N\in w+H\mid B=b,T=t]
 =\frac{Z_{w+H}(\alpha)}{Z_{w+C_A}(\alpha)}
 \le\frac{Z_H(\alpha)}{Z_{C_A}(\alpha)}.
\end{equation}
A randomized decoder averages the success probabilities of its possible guesses and therefore cannot beat a most probable coset. An optimal decoder chooses such a coset. The bound is independent of $b$; averaging this maximum with weights $\Pr[B=b\mid T=t]$ therefore proves
\begin{equation}
 \pg(K\mid B,T=t)\le\frac{Z_H(\alpha)}{Z_{C_A}(\alpha)}. \label{eq:bob-transcript}
\end{equation}
This step uses only classical probabilities and requires no min-entropy characterization.

\emph{Comparison.} Apply Eq.~\eqref{eq:subspace} to $D,H\subseteq C_A$. Combining the bounds gives
\begin{equation}
 \pg(K\mid E,T=t)
 \ge\frac{Z_J}{Z_D}
 \ge\frac{Z_H}{Z_{C_A}}
 \ge\pg(K\mid B,T=t), \label{eq:chain}
\end{equation}
as required. All quantities in the middle are evaluated at $\alpha$. The physical attack was chosen before fixing the protocol or transcript, completing the proof.
\end{proof}

\section{Proofs of Lemmas~\ref{lem:attack} and~\ref{lem:spectrum}}\label{sec:constructions}
We now give the proofs of Lemmas~\ref{lem:attack} and~\ref{lem:spectrum}, which supply the physical realization of the attack and the spectral identity used in Sec.~\ref{sec:proof}.

\subsection{Physical realization of the attack}\label{sec:attack-realization}
\begin{proof}[Proof of Lemma~\ref{lem:attack}]
Let $X,Z$ denote the Pauli operators. Eve applies the isometry
\begin{equation}
 V_p\ket\psi=\sum_{e,z\in\F}\sqrt{q_{e,z}}\,
 X^eZ^z\ket\psi_B\ket{e,z}_{E_0},\qquad
 q_{0,0}=1-\frac{3p}{2},\quad q_{e,z}=\frac p2\ \ ((e,z)\ne(0,0)).
 \label{eq:isometry}
\end{equation}
The coefficients are nonnegative and sum to one for the range considered. This is a depolarizing Pauli channel: either of the two Pauli errors that anticommute with a given measurement basis flips its outcome, giving QBER $p$ in each of the three bases. The interaction does not require knowledge of that basis.

For matched-$Z$ signals, conditional on Alice's bit $a$ and Bob's bit $a+e$, Eve's normalized states are
\begin{align}
 \ket{\chi_a^{(0)}}&=
 \sqrt{\frac{1-3p/2}{1-p}}\ket{0,0}
 +(-1)^a\sqrt{\frac{p/2}{1-p}}\ket{0,1},\\
 \ket{\chi_a^{(1)}}&=\frac{\ket{1,0}+(-1)^a\ket{1,1}}{\sqrt2}.
\end{align}
The two error subspaces are orthogonal, so Eve can record $e$ without disturbing either pair. Within each pair the overlaps are respectively $\beta=(1-2p)/(1-p)$ and zero.

For completeness, two pure states of real overlap $b\in[0,\gamma]$ can be mapped by a quantum channel to two pure states of overlap $\gamma>0$. Choose environment states $\ket{\eta_0},\ket{\eta_1}$ with overlap $b/\gamma$. The mapping
\begin{equation}
 \ket{\chi_a}\longmapsto\ket{\phi_a}\ket{\eta_a}
\end{equation}
preserves the Gram matrix and hence defines an isometry on the input span. Tracing out the environment implements the desired channel. Equation~\eqref{eq:overlap} allows this construction for the no-error pair; the orthogonal error pair allows it as well. Eve uses these channels conditioned on $e$, retains the error flag, and discards their environments. The result is Eq.~\eqref{eq:attack}. Repeating the same interaction independently gives a collective attack for every raw length. The local degradation leaves all three observed QBERs unchanged.
\end{proof}

\subsection{Classical representation of the spectrum}\label{sec:spectrum-proof}
\begin{proof}[Proof of Lemma~\ref{lem:spectrum}]
The states in Eq.~\eqref{eq:average-state} are indexed by $a=G^{\mathsf T}x$, $x\in\F^k$. The density operator $\rho_C$ and its normalized Gram matrix have the same nonzero eigenvalues, including multiplicities. The latter matrix has entries
\begin{equation}
 \mathcal G_C(x,y)=2^{-k}\gamma^{\wt(G^{\mathsf T}(x+y))}. \label{eq:gram}
\end{equation}
The binary characters $2^{-k/2}(-1)^{s\cdot x}$ diagonalize this matrix. The eigenvalue indexed by $s$ is
\begin{equation}
 \lambda_C(s)=2^{-k}\sum_{x\in\F^k}(-1)^{s\cdot x}
 \gamma^{\wt(G^{\mathsf T}x)}. \label{eq:eigenvalues}
\end{equation}
If $W_j$ are independent Bernoulli$((1-\gamma)/2)$ variables, then
\begin{equation}
 \mathbb E(-1)^{x\cdot GW}
 =\prod_{j=1}^n\mathbb E(-1)^{(G^{\mathsf T}x)_jW_j}
 =\gamma^{\wt(G^{\mathsf T}x)}.
\end{equation}
Fourier inversion identifies Eq.~\eqref{eq:eigenvalues} with $\Pr[GW=s]$. The Gram matrix has dimension $2^k$, whereas $\rho_C$ acts on a space of dimension $2^n$; adding $2^n-2^k$ zero eigenvalues therefore gives the full spectrum of $\rho_C$. Shannon and collision entropies are unaffected by these zeros, proving Eq.~\eqref{eq:spectrum-entropies}.
\end{proof}

\section{Conclusion}\label{sec:conclusion}
The finite-length comparison identifies a limitation of adaptive linear processing at the known threshold. Its main ingredients are a deliberately weakened quantum attack and the geometry of the full public transcript. The information-theoretic comparison, rather than any restriction on decoding efficiency, prevents reliable linear key agreement from being secret. Extending the comparison beyond linear public communication and linear candidates remains a question for future work.

\section*{Acknowledgments}
OpenAI's ChatGPT was used under the author's direction to assist with mathematical analysis, checking derivations, manuscript organization, and drafting in English and Japanese. The author guided this assistance through iterative technical questions, examined intermediate arguments, and revised the manuscript.

\section*{Data availability}
This study is analytical. The derivations use no external datasets.

\appendix
\clearpage
\section{An elementary R\'enyi comparison via conditional collision probabilities}
\label{app:detailed-proof}

Equation~\eqref{eq:renyi} follows from Ref.~\cite[Theorem~1]{Vijay}.
For readers wishing to verify the comparison independently, this appendix gives
a self-contained proof using averaged conditional collision probabilities.
The goal is to show that, for two nested linear images of independent Bernoulli
variables, the difference of their Shannon entropies is at least the difference
of their collision entropies.

The proof has two comparisons. First, we bound the averaged conditional collision
probability from above by replacing one bit at a time with either a bit fixed
at zero or a uniform bit. Second, we compare the resulting bound with a ratio
of collision probabilities for two independent copies. In the one-bit calculation,
we first treat the case where the distribution of the conditioning variable
does not change; when that distribution changes as well, we compare the
denominators. We also prove the positive-correlation inequality used in the
second comparison.

As in the main text, all logarithms are base two, and all vectors, matrices,
and additions are over $\F$. Thus addition of bits means exclusive OR.
The variable $T$ in this appendix is an auxiliary linear image, not the public
transcript of the protocol. Likewise, the probabilities $p_j$ below describe
auxiliary Bernoulli variables, not the QBER. Their identification with the
main-text notation is given in Appendix~\ref{collision:sec:application}.

\subsection{The quantities to be compared}
\label{collision:sec:goal}

\subsubsection{Setting and goal}

Let $Z=(Z_1,\ldots,Z_n)^{\mathsf T}$ have independent components with
$\Prob(Z_j=1)=p_j$. The parameters $p_j\in[0,1]$ need not be equal.
Let binary matrices $L_1,L_2$ satisfy
\[
 \ker L_1\subseteq\ker L_2,
\]
and put
\begin{equation}
 T=L_1Z,\qquad V=L_2Z.
 \label{collision:eq:TV}
\end{equation}
The kernel inclusion means that knowing $T$ determines $V$.
Indeed, if $L_1z=L_1z'$, then
$z+z'\in\ker L_1\subseteq\ker L_2$, and hence $L_2z=L_2z'$.
This correspondence is linear, so for a suitable matrix $M$ we have
\begin{equation}
 V=MT,\qquad L_2=ML_1.
 \label{collision:eq:M}
\end{equation}
We view $T$ as the finer variable and $V$ as the variable obtained by discarding
some of its information.

Write the distributions and collision entropy as
\[
 P(t)=\Prob(T=t),\quad Q(v)=\Prob(V=v),\qquad
 H_2(X)=-\log_2\sum_x\Prob(X=x)^2.
\]
Our goal is
\begin{equation}
 \boxed{H(T\mid V)=H(T)-H(V)\ge H_2(T)-H_2(V).}
 \label{collision:eq:goal}
\end{equation}
Both sides are \emph{differences} of entropies. Consequently, applying
$H(X)\ge H_2(X)$ separately to $T$ and $V$ does not establish
Eq.~\eqref{collision:eq:goal}.

\subsubsection{An intermediate averaged conditional collision probability}

For $Q(v)>0$, define
\[
 c_v:=\sum_{t:Mt=v}\Prob(T=t\mid V=v)^2.
\]
For $Q(v)=0$, set $c_v=0$ by convention.
This is the probability that two independent draws of $T$ from its conditional
distribution given a positive-probability value $V=v$ coincide. Average it using the actual probability
$Q(v)$ of the conditioning value:
\begin{equation}
 \mathcal C:=\sum_{v:Q(v)>0}Q(v)c_v
 =\sum_{t:P(t)>0}\frac{P(t)^2}{Q(Mt)}.
 \label{collision:eq:C}
\end{equation}
No denominator in this sum is zero, because $P(t)>0$ implies $Q(Mt)>0$.

The quantity corresponding to the right-hand side of our goal is instead
\begin{equation}
 \mathcal J:=\frac{\sum_tP(t)^2}{\sum_vQ(v)^2}
 =2^{-[H_2(T)-H_2(V)]}
 =\frac{\sum_vQ(v)^2c_v}{\sum_vQ(v)^2}.
 \label{collision:eq:J}
\end{equation}
These two quantities average the same $c_v$, but $\mathcal C$ uses weights
$Q(v)$, whereas $\mathcal J$ uses normalized weights $Q(v)^2$.
Because of this difference, $\mathcal C\le\mathcal J$ need not hold for
general distributions. The substance of the proof is to establish this
direction of comparison using independence and linearity.

\begin{theorem}[A comparison implying Eq.~\eqref{eq:renyi}]
\label{collision:thm:main}
In the above setting,
\begin{equation}
 \boxed{H(T\mid V)\ge-\log_2\mathcal C
 \ge H_2(T)-H_2(V).}
 \label{collision:eq:strong}
\end{equation}
In particular, Eq.~\eqref{collision:eq:goal} holds.
\end{theorem}

The first inequality holds for general distributions and can be checked
immediately. For a realized pair $(T,V)=(t,v)$, put
$A(t,v)=\Prob(T=t\mid V=v)$. Then
\[
 H(T\mid V)=\E[-\log_2 A(T,V)],\qquad
 \E A(T,V)=\mathcal C.
\]
Convexity of $-\log_2 x$ and Jensen's inequality give
\begin{equation}
 H(T\mid V)\ge-\log_2\mathcal C,
 \qquad\text{or equivalently}\qquad 2^{-H(T\mid V)}\le\mathcal C.
 \label{collision:eq:jensen}
\end{equation}
It remains to prove $\mathcal C\le\mathcal J$.

\subsection{An intermediate quantity: retaining only uniform coordinates}
\label{collision:sec:rank}

For each coordinate, put
\begin{equation}
 \theta_j:=4p_j(1-p_j),
 \label{collision:eq:theta}
\end{equation}
and consider a random subset $S\subseteq[n]$ that includes each $j$
independently with probability $\theta_j$, where $[n]=\{1,\ldots,n\}$.
Write its distribution as
\begin{equation}
 \mu(S)=\prod_{j\in S}\theta_j\prod_{j\notin S}(1-\theta_j).
 \label{collision:eq:mu}
\end{equation}
The set $S$ specifies which coordinates will be made uniform. At this stage,
we are not constructing a random variable with the same distribution as $Z$.

Define the coordinate subspace
\[
 \F^S=\{z\in\F^n:z_j=0\text{ for }j\notin S\}.
\]
Let $L_{i,S}$ be the matrix consisting of the columns of $L_i$ indexed by $S$.
We use the following difference of dimensions:
\begin{equation}
 r_i(S):=\rank L_{i,S},\qquad
 \kappa(S):=r_1(S)-r_2(S).
 \label{collision:eq:kappa}
\end{equation}

\begin{lemma}[Remaining degrees of freedom for uniform input]
\label{collision:lem:rank}
If the input is uniform on $\F^S$, then $T$ is uniform on $2^{r_1(S)}$ values,
and $V$ is uniform on $2^{r_2(S)}$ values. Moreover, for every $v$ of positive
probability, $T\mid V=v$ is uniform on $2^{\kappa(S)}$ values.
Writing $\mathcal C_S$ for $\mathcal C$ with this uniform input, we have
\begin{equation}
 \mathcal C_S=2^{-\kappa(S)}.
 \label{collision:eq:uniformC}
\end{equation}
Furthermore, $S\subseteq S'$ implies $\kappa(S)\le\kappa(S')$.
\end{lemma}

\begin{proof}
Each point in the image of a linear map has as many preimages as there are
elements in its kernel. Thus the image of a uniform input is uniform.
Write $\mathcal T_S=L_1(\F^S)$ and $\mathcal V_S=L_2(\F^S)$ for the image
spaces of $T$ and $V$. The map $M:\mathcal T_S\to\mathcal V_S$ is surjective,
so its kernel has dimension $r_1(S)-r_2(S)=\kappa(S)$.
The preimage of each $v$ is a translate of this kernel: it consists of one
preimage $t_v$ plus all kernel elements. It therefore contains
$2^{\kappa(S)}$ points, and the conditional distribution is uniform on them.

For monotonicity, rather than directly comparing two differences of ranks,
we use
\begin{equation}
 \kappa(S)=\dim L_1(\F^S\cap\ker L_2).
 \label{collision:eq:kappa-image}
\end{equation}
The space on the right consists of all changes in $T$ obtained by varying
only the input coordinates in $S$ while leaving $V$ unchanged. It is exactly
the kernel of $M|_{\mathcal T_S}$. If $S\subseteq S'$, then
$\F^S\cap\ker L_2\subseteq\F^{S'}\cap\ker L_2$, so the corresponding
images are also nested. Taking dimensions proves the claim.
\end{proof}

The full proof establishes the chain
\begin{equation}
 \boxed{2^{-H(T\mid V)}\le\mathcal C
 \le\E_\mu 2^{-\kappa(S)}\le\mathcal J.}
 \label{collision:eq:roadmap}
\end{equation}
The first inequality is Eq.~\eqref{collision:eq:jensen}. We prove the remaining
two in the following order.
\begin{enumerate}
\item Appendix~\ref{collision:sec:onebit} establishes a one-bit bound.
Iterating it over all bits in Appendix~\ref{collision:sec:iteration} gives
$\mathcal C\le\E_\mu2^{-\kappa(S)}$.
\item Appendix~\ref{collision:sec:comparison} proves
$\E_\mu2^{-\kappa(S)}\le\mathcal J$ in two stages: first we obtain a
representation of $\mathcal J$, and then we use it to compare the two averages.
\end{enumerate}
Finally, Appendix~\ref{collision:sec:completion} combines the inequalities to
complete the proof of Theorem~\ref{collision:thm:main}.

\subsection{The one-bit bound: starting with an unchanged conditioning distribution}
\label{collision:sec:onebit}

This subsection proves the one-bit bound needed for the middle inequality
in Eq.~\eqref{collision:eq:roadmap}. We will iterate it over all bits in the
next subsection. The quantity to be bounded remains $\mathcal C$ from
Eq.~\eqref{collision:eq:C}, and the random variables remain $T=L_1Z$ and
$V=L_2Z$.

Fix the distribution of every bit except the $j$th, as well as $L_1,L_2$,
and vary only $\Prob(Z_j=1)=p$. Denote the resulting distributions of $T,V$
by $P_p(t),Q_p(v)$, and put
\begin{equation}
 \mathcal C(p):=\sum_{t\in\operatorname{im}L_1}
                  \frac{P_p(t)^2}{Q_p(Mt)}.
 \label{collision:eq:Cp}
\end{equation}
Since $P_p(t)\le Q_p(Mt)$, a zero denominator also has a zero numerator;
we assign that summand the value zero. This convention lets us use the same
summation set for every $p$.
The distributions $P_0,Q_0$ are those obtained by \emph{replacing} the $j$th
bit with zero. If the original $p_j<1$, independence implies that this is
the same as conditioning the original distribution on $Z_j=0$.
For $p_j=1$, that conditioning is undefined, but replacing the $j$th bit
with zero remains well-defined. Similarly, $\mathcal C(0)$ and
$\mathcal C(1/2)$ keep the other bits unchanged and replace only the $j$th
bit with a fixed zero or a uniform bit, respectively.

\begin{lemma}[One-bit endpoint bound]\label{collision:lem:onebit}
If the $j$th bit is independent of all the other bits, then
\begin{equation}
 \boxed{\mathcal C(p)\le(1-\theta)\mathcal C(0)
             +\theta\mathcal C(1/2),\qquad\theta=4p(1-p).}
 \label{collision:eq:onebit}
\end{equation}
The other bits may be mutually correlated in this lemma.
\end{lemma}

\begin{proof}
\subsubsection*{Preparation: working with the original probability distributions}

Let $e_j$ be the standard basis vector whose only nonzero entry is its $j$th,
and put
\[
 \delta=L_1e_j,\qquad\gamma=L_2e_j=M\delta.
\]
Here $\gamma$ is a temporary vector notation used only in this one-bit proof;
it is unrelated to the scalar overlap in the main text.
These vectors are the changes in $T,V$ when $Z_j$ changes from zero to one.
Independence of the $j$th bit from the other bits gives
\begin{align}
 P_p(t)&=(1-p)P_0(t)+pP_0(t+\delta),\notag\\
 Q_p(v)&=(1-p)Q_0(v)+pQ_0(v+\gamma).
 \label{collision:eq:shifts}
\end{align}
Thus the contribution from $Z_j=1$ is the original distribution shifted by
$\delta$ or $\gamma$, respectively. If $\delta=0$, then also $\gamma=0$,
both distributions remain unchanged, and the lemma holds with equality.
Henceforth assume $\delta\ne0$.

Partition $\operatorname{im}L_1$, including points of zero probability, into
pairs $\{t,t+\delta\}$. For one such pair, write $v=Mt$; the value of $V$
corresponding to the other point is $M(t+\delta)=v+\gamma$.
Fix this pair and define its two-point contribution to $\mathcal C(p)$ by
\begin{equation}
 F(p):=\frac{P_p(t)^2}{Q_p(v)}+
       \frac{P_p(t+\delta)^2}{Q_p(v+\gamma)}.
 \label{collision:eq:pair-definition}
\end{equation}
The quantity $F(p)$ refers to this fixed pair in both cases below.
Summing each pair's $F(p)$ exactly once gives $\mathcal C(p)$.

\subsubsection*{Case 1: the distribution of $V$ is unchanged ($\gamma=0$)}

Here $Q_p=Q_0$, and both denominators in
Eq.~\eqref{collision:eq:pair-definition} are $Q_0(v)$.
For a pair with $Q_0(v)>0$, substituting Eq.~\eqref{collision:eq:shifts}
and expanding the squares gives
\begin{align}
 F(p)&=\frac{P_p(t)^2+P_p(t+\delta)^2}{Q_0(v)}\notag\\
 &=(1-\theta)\frac{P_0(t)^2+P_0(t+\delta)^2}{Q_0(v)}\notag\\
 &\quad+\theta\frac{[P_0(t)+P_0(t+\delta)]^2}{2Q_0(v)}.
 \label{collision:eq:unchanged}
\end{align}
\begin{samepage}
The first fraction on the final right-hand side is $F(0)$. Moreover,
Eq.~\eqref{collision:eq:shifts} gives
\[
 P_{1/2}(t)=P_{1/2}(t+\delta)
 =\frac{P_0(t)+P_0(t+\delta)}{2},
\]
so the second fraction is $F(1/2)$.
\end{samepage}
Thus this pair satisfies $F(p)=(1-\theta)F(0)+\theta F(1/2)$.
If $Q_0(v)=0$, then $F(p)=0$ for every $p$, and the same equality holds.
Summing over all pairs therefore yields
\[
 \mathcal C(p)=(1-\theta)\mathcal C(0)+\theta\mathcal C(1/2).
\]
When the denominators remain unchanged, the desired bound is an equality.

\subsubsection*{Case 2: the two values on the $V$ side are also exchanged ($\gamma\ne0$)}

Now $v$ and $v+\gamma$ are distinct, and both denominators vary with $p$.
Fix one pair and, only for the calculation that follows, abbreviate the four
original probabilities as
\[
 a=P_0(t),\quad b=P_0(t+\delta),\qquad
 c=Q_0(v),\quad d=Q_0(v+\gamma).
\]
Since $V=MT$ and $v=Mt$, the event $T=t$ implies $V=v$, and the event
$T=t+\delta$ implies $V=v+\gamma$. These event inclusions also hold with
the $j$th bit fixed at zero, giving $a\le c$ and $b\le d$.
Equations~\eqref{collision:eq:pair-definition} and~\eqref{collision:eq:shifts}
give
\begin{equation}
 F(p)=\frac{[(1-p)a+pb]^2}{(1-p)c+pd}
       +\frac{[(1-p)b+pa]^2}{(1-p)d+pc}.
 \label{collision:eq:two-terms}
\end{equation}
First assume $c,d>0$, so both denominators are positive even at $p=0,1$.
Putting fractions over a common denominator gives the identity
\[
 \frac{a^2}{c}+\frac{b^2}{d}
 =\frac{(a+b)^2}{c+d}+\frac{(ad-bc)^2}{cd(c+d)}.
\]
Apply this identity to the four mixed probabilities appearing in
Eq.~\eqref{collision:eq:two-terms}. The two probabilities that are squared
sum to $a+b$, the denominators sum to $c+d$, and the difference of cross
products is
\[
 [(1-p)a+pb][(1-p)d+pc]
 -[(1-p)b+pa][(1-p)c+pd]
 =(1-2p)(ad-bc).
\]
Consequently,
\begin{equation}
 F(p)=\frac{(a+b)^2}{c+d}
 +\frac{(1-2p)^2(ad-bc)^2}
 {(c+d)[(1-p)c+pd][(1-p)d+pc]}.
 \label{collision:eq:pair-contribution}
\end{equation}
The first term is independent of $p$, so only the second needs to be bounded.
The product of the two mixed denominators satisfies
\begin{equation}
 [(1-p)c+pd][(1-p)d+pc]
 =cd+p(1-p)(c-d)^2\ge cd.
 \label{collision:eq:denominator}
\end{equation}
The first term in Eq.~\eqref{collision:eq:pair-contribution} is $F(1/2)$,
and the second term at $p=0$ is $F(0)-F(1/2)$.
Since $(1-2p)^2=1-\theta$, Eq.~\eqref{collision:eq:denominator} gives
\begin{align}
 F(p)
 &\le F(1/2)+(1-\theta)[F(0)-F(1/2)]\notag\\
 &=(1-\theta)F(0)+\theta F(1/2).
 \label{collision:eq:pair-bound}
\end{align}
This is the inequality corresponding to the equality in Case~1.

It remains to check zero original denominators.
If $c=0<d$, then $a=0$, and for $0<p<1$ the two terms in
Eq.~\eqref{collision:eq:two-terms} are $pb^2/d$ and $(1-p)b^2/d$.
Their sum is $b^2/d$, also at $p=0,1$ after omitting zero-probability terms.
Similarly, if $d=0<c$, the contribution is always $a^2/c$.
If $c=d=0$, then $a=b=0$, and the contribution is always zero.
Thus the endpoint bound is an equality for each of these pairs.

Summing Eq.~\eqref{collision:eq:pair-bound} and the equalities for
zero-probability cases over all pairs gives Eq.~\eqref{collision:eq:onebit}.
This completes both cases of the proof.
\end{proof}

The one-bit argument uses the fact that adding the independent $j$th bit
exchanges two points of $T$ and induces the corresponding addition on $V$.
It also includes the case where $V$ does not change.
In the next subsection, we use this bound to replace coordinates successively
with fixed zeros or uniform bits, obtaining the middle inequality in
Eq.~\eqref{collision:eq:roadmap}.

\subsection{Iterating over all bits}
\label{collision:sec:iteration}

The preceding subsection varied only one bit. We now iterate the same bound
over all bits, expressing an upper bound on $\mathcal C$ in terms of uniform
inputs. This is where we use independence of $Z_1,\ldots,Z_n$.
Apply Lemma~\ref{collision:lem:onebit} to the first bit, then apply it to the
second bit in each of the resulting two terms, and continue through the $n$th
bit. Replacing a bit with a fixed zero or an independent uniform bit preserves
its independence from the other bits, so the lemma applies at every step.

For example, with two bits and explicit probability arguments,
\begin{align*}
 \mathcal C(p_1,p_2)
 &\le(1-\theta_1)\mathcal C(0,p_2)
       +\theta_1\mathcal C(1/2,p_2)\\
 &\le(1-\theta_1)(1-\theta_2)\mathcal C(0,0)
       +(1-\theta_1)\theta_2\mathcal C(0,1/2)\\
 &\quad+\theta_1(1-\theta_2)\mathcal C(1/2,0)
       +\theta_1\theta_2\mathcal C(1/2,1/2).
\end{align*}
The same expansion applies for general $n$. If $S$ is the set of coordinates
made uniform, the coefficient of its term is precisely $\mu(S)$ from
Eq.~\eqref{collision:eq:mu}. Each term has input uniform on $\F^S$, so
Lemma~\ref{collision:lem:rank} gives
\begin{equation}
 \boxed{\mathcal C\le\sum_{S\subseteq[n]}\mu(S)\mathcal C_S
       =\E_\mu 2^{-\kappa(S)}.}
 \label{collision:eq:C-bound}
\end{equation}
This is the middle inequality in Eq.~\eqref{collision:eq:roadmap}.
The argument does not claim that the original distribution \emph{equals}
a mixture of these uniform distributions. It applies an upper bound on
$\mathcal C$ one bit at a time.

\subsection{Comparing the two averages}
\label{collision:sec:comparison}

Equation~\eqref{collision:eq:C-bound} gives
$\mathcal C\le\E_\mu2^{-\kappa(S)}$. We now prove the final inequality in
Eq.~\eqref{collision:eq:roadmap},
\[
 \E_\mu2^{-\kappa(S)}\le\mathcal J.
\]
The proof has two stages. First, we express $\mathcal J$ using the subset
distribution $\mu$. Second, we apply a positive-correlation inequality to
that representation.

\subsubsection{Stage 1: expressing the target quantity using independent copies}
\label{collision:sec:copies}

We first view $\mathcal J$ as a ratio of collision probabilities for two
independent copies, and obtain a representation using $\mu$.

Let $Z'$ be an independent copy of $Z$ and put $D=Z+Z'$.
Then $T'=L_1Z'$ and $V'=L_2Z'$ form an independent copy $(T',V')$ of $(T,V)$;
within this copy, $V'=MT'$. By linearity,
\[
 T=T'\ \Longleftrightarrow\ L_1D=0,\qquad
 V=V'\ \Longleftrightarrow\ L_2D=0.
\]
The probability that two independent copies coincide is the sum of squared
probabilities. Thus Eq.~\eqref{collision:eq:J} becomes
\begin{equation}
 \mathcal J=\frac{\Prob(L_1D=0)}{\Prob(L_2D=0)}.
 \label{collision:eq:J-copy}
\end{equation}
Since $T=T'$ implies $V=V'$, it also has the interpretation
\[
 \mathcal J=\Prob(T=T'\mid V=V').
\]
For $\mathcal C$, one first selects a common $v$ with probability $Q(v)$,
then independently draws two values of $T$ from the conditional distribution.
For $\mathcal J$, one first draws $(T,V)$ and $(T',V')$ independently, then
conditions on $V=V'$. This difference produces the $Q(v)^2$ weights in
Eq.~\eqref{collision:eq:J}.

The coordinates $D_j$ are independent, with
\begin{equation}
 \Prob(D_j=1)=p_j(1-p_j)+(1-p_j)p_j
       =2p_j(1-p_j)=\frac{\theta_j}{2}.
 \label{collision:eq:D-law}
\end{equation}
The distribution of $D$ can therefore be generated \emph{exactly} as follows.
\begin{samepage}
\begin{enumerate}
\item Select each coordinate independently with probability $\theta_j$,
and let $S$ be the set of selected coordinates.
\item Assign independent uniform bits to coordinates in $S$, and zeros
to coordinates outside $S$.
\end{enumerate}
\end{samepage}
Indeed, each coordinate is one with probability $\theta_j\times\tfrac12$,
and the coordinates have the required independence.
This is the \emph{same} subset distribution $\mu$ that appears in the bound
of Appendix~\ref{collision:sec:iteration}.

For fixed $S$, the input $D$ is uniform on $\F^S$. Uniformity of the image,
as in Lemma~\ref{collision:lem:rank}, gives
\[
 \Prob(L_iD=0\mid S)=2^{-r_i(S)}.
\]
Consequently,
\begin{equation}
 \mathcal J=\frac{\E_\mu 2^{-r_1(S)}}{\E_\mu 2^{-r_2(S)}}
       =\frac{\E_\mu[a(S)b(S)]}{\E_\mu a(S)},
 \quad a(S):=2^{-r_2(S)},\quad b(S):=2^{-\kappa(S)}.
 \label{collision:eq:J-rank}
\end{equation}
The last equality uses $r_1=r_2+\kappa$.
The first ratio in Eq.~\eqref{collision:eq:J-rank} is already the desired
quantity $2^{-[H_2(T)-H_2(V)]}$. The upper bound obtained in
Eq.~\eqref{collision:eq:C-bound}, however, is
\[
 \E_\mu2^{-\kappa(S)}
 =\E_\mu\!\left[\frac{2^{-r_1(S)}}{2^{-r_2(S)}}\right].
\]
It is an \emph{average of ratios}, whereas Eq.~\eqref{collision:eq:J-rank}
is a \emph{ratio of averages}; these need not be equal.

In the notation $a(S),b(S)$, our upper bound is $\E_\mu b(S)$ and the target
is $\E_\mu[a(S)b(S)]/\E_\mu a(S)$. The latter is an average of $b(S)$
reweighted by $a(S)$. This puts the two quantities in comparable form.
The next stage shows that the reweighting cannot decrease the average.

\subsubsection{Stage 2: comparing the averages by positive correlation}
\label{collision:sec:association}

By Eq.~\eqref{collision:eq:J-rank}, it suffices to show
$\E_\mu[ab]\ge\E_\mu a\,\E_\mu b$.
We prove this positive correlation using the fact that both $a(S)$ and $b(S)$
are decreasing with respect to inclusion of $S$.

Adding coordinates to $S$ cannot decrease $r_2(S)$.
By Lemma~\ref{collision:lem:rank}, it cannot decrease $\kappa(S)$ either.
Thus both $a(S)$ and $b(S)$ in Eq.~\eqref{collision:eq:J-rank} are decreasing
with respect to inclusion. The needed fact is Harris's inequality for product
measures~\cite{Harris}. We include its inductive proof to keep this appendix
self-contained.

\begin{lemma}[Positive correlation under independent selection]
\label{collision:lem:association}
Let $S$ be a random set obtained by selecting each coordinate independently.
If real-valued functions $a,b$ are both decreasing with respect to inclusion,
then
\begin{equation}
 \E[a(S)b(S)]\ge\E a(S)\,\E b(S).
 \label{collision:eq:association}
\end{equation}
\end{lemma}

\begin{proof}
We induct on the number of coordinates. With no coordinates, $S$ is
deterministic and equality holds. A coordinate selected with probability
zero or one may be fixed to its deterministic value, reducing the number
of coordinates. We may therefore suppose that the final coordinate is
selected with probability $0<q<1$.

Let $J\in\{0,1\}$ indicate whether this coordinate is selected, and put
\[
 A_i:=\E[a\mid J=i],\qquad B_i:=\E[b\mid J=i]
 \qquad(i=0,1).
\]
Conditioning on $J$ preserves independence of the remaining coordinates,
so the induction hypothesis gives
\[
 \E[ab]\ge(1-q)A_0B_0+qA_1B_1.
\]
On the other hand,
\[
 \E a\,\E b=[(1-q)A_0+qA_1][(1-q)B_0+qB_1].
\]
Subtracting the latter expression from the preceding right-hand side and
expanding gives
\begin{equation}
 q(1-q)(A_0-A_1)(B_0-B_1).
 \label{collision:eq:association-difference}
\end{equation}
The remaining coordinates have the same distribution for $J=0$ and $J=1$,
and adding the final coordinate cannot increase either function.
Hence $A_0\ge A_1$ and $B_0\ge B_1$.
Expression~\eqref{collision:eq:association-difference} is nonnegative,
completing the induction.
\end{proof}

Applying the lemma to Eq.~\eqref{collision:eq:J-rank} and using
$\E_\mu a>0$, we obtain
\begin{equation}
 \mathcal J
 =\frac{\E_\mu[ab]}{\E_\mu a}
 \ge\frac{\E_\mu a\,\E_\mu b}{\E_\mu a}
 =\E_\mu 2^{-\kappa(S)}.
 \label{collision:eq:J-bound}
\end{equation}
This is the final inequality in Eq.~\eqref{collision:eq:roadmap}.

\subsection{Completion of the proof of Theorem~\ref{collision:thm:main}}
\label{collision:sec:completion}

\begin{proof}
Combining Eqs.~\eqref{collision:eq:jensen}, \eqref{collision:eq:C-bound},
and~\eqref{collision:eq:J-bound} gives
\begin{equation}
 2^{-H(T\mid V)}\le\mathcal C
 \le\E_\mu 2^{-\kappa(S)}
 \le\mathcal J=2^{-[H_2(T)-H_2(V)]}.
 \label{collision:eq:complete}
\end{equation}
Taking negative logarithms reverses the inequalities and proves
Eq.~\eqref{collision:eq:strong}. The argument includes bit probabilities
zero, $1/2$, and one.
\end{proof}

\subsection{A two-bit example of all the quantities}
\label{collision:sec:example}

For independent $Z_1,Z_2\sim\Bern(1/4)$, let
\[
 T=(Z_1,Z_2),\qquad V=Z_1+Z_2.
\]
The variable $V$ retains only the parity of the two bits. The distribution is
\begin{center}
\begin{tabular}{cccc}
\toprule
$t$ & $P(t)$ & $v$ & $\Prob(T=t\mid V=v)$\\
\midrule
$00$ & $9/16$ & 0 & $9/10$\\
$11$ & $1/16$ & 0 & $1/10$\\
$01$ & $3/16$ & 1 & $1/2$\\
$10$ & $3/16$ & 1 & $1/2$\\
\bottomrule
\end{tabular}
\end{center}
Here $Q(0)=5/8$ and $Q(1)=3/8$.
Writing $h_2(u)=-u\log_2u-(1-u)\log_2(1-u)$ for the binary entropy,
we have
\begin{align}
 H(T\mid V)&=\frac58 h_2(1/10)+\frac38
           \simeq0.668122,\notag\\
 2^{-H(T\mid V)}&\simeq0.629325,\notag\\
 \mathcal C&=\frac58\left[\left(\frac9{10}\right)^2+
                    \left(\frac1{10}\right)^2\right]
                  +\frac38\left[2\left(\frac12\right)^2\right]
             =\frac7{10}.
 \label{collision:eq:exampleC}
\end{align}

Since $\theta_1=\theta_2=3/4$, the subsets and their dimensions are
\begin{center}
\begin{tabular}{ccccc}
\toprule
$S$ & $\mu(S)$ & $r_1(S)$ & $r_2(S)$ & $\kappa(S)$\\
\midrule
$\varnothing$ & $1/16$ & 0 & 0 & 0\\
$\{1\}$ & $3/16$ & 1 & 1 & 0\\
$\{2\}$ & $3/16$ & 1 & 1 & 0\\
$\{1,2\}$ & $9/16$ & 2 & 1 & 1\\
\bottomrule
\end{tabular}
\end{center}
If only one bit can vary, its value is determined by the parity $V$, leaving
no degree of freedom. If both can vary, simultaneous flipping by $(1,1)$
preserves the parity, leaving one degree of freedom. Therefore
\[
 \E_\mu 2^{-\kappa(S)}=\frac7{16}+\frac9{16}\frac12=\frac{23}{32}.
\]
Also,
\[
 \sum_tP(t)^2=\frac{25}{64},\qquad
 \sum_vQ(v)^2=\frac{17}{32},\qquad
 \mathcal J=\frac{25}{34}.
\]
Numerically, the chain proved above is
\begin{equation}
 \underbrace{0.629325\ldots}_{2^{-H(T\mid V)}}
 <\underbrace{0.7}_{\mathcal C}
 <\underbrace{0.71875}_{\E_\mu2^{-\kappa(S)}}
 <\underbrace{0.735294\ldots}_{\mathcal J}.
 \label{collision:eq:example-chain}
\end{equation}
Every step is strict in this example. In particular,
$\E_\mu2^{-\kappa(S)}=23/32$ differs from $\mathcal J=25/34$, illustrating
why the comparison of averages in Appendix~\ref{collision:sec:association}
is needed.
We also have $c_0=41/50>c_1=1/2$. The weights on $v=0,1$ change from
$(5/8,3/8)$ for $\mathcal C$ to $(25/34,9/34)$ for $\mathcal J$.
The greater weight on $v=0$, where collisions are more likely, directly
explains $\mathcal J>\mathcal C$. The final entropy comparison is
\[
 H(T\mid V)\simeq0.668122
 \ >\ H_2(T)-H_2(V)=\log_2\frac{34}{25}\simeq0.443607.
\]

\subsection{Application to the main text and the role of each step}
\label{collision:sec:application}

To apply the result in the main-text notation, set
\[
 L_1=G,\qquad L_2=G_0,\qquad Z=W,\qquad
 p_j=\frac{1-\gamma}{2}.
\]
Here $\gamma$ is the scalar overlap in Eq.~\eqref{eq:alphabet}.
The rows of $G_0$ are the final rows of $G$, so
$\ker G\subseteq\ker G_0$. Thus the variables of this appendix become
$T=GW$ and $V=G_0W$, and the spectral identification in
Lemma~\ref{lem:spectrum} gives the eigenvalue lists of $\rho_D$ and
$\rho_{D_0}$, respectively. Applying Eq.~\eqref{collision:eq:goal} to these
lists proves Eq.~\eqref{eq:renyi}. In this application,
$\theta_j=1-\gamma^2$. The random set $S$ in this appendix specifies the
coordinates made \emph{uniform}.

The uses of independence and linearity can be summarized as follows.
\begin{itemize}
\item Adding one bit mixes the original probabilities $P_p,Q_p$ in pairs.
If the distribution of $V$ is unchanged, expanding the squares gives the
endpoint bound with equality. If the values on the $V$ side are exchanged
as well, Eq.~\eqref{collision:eq:pair-contribution} and the denominator
comparison in Eq.~\eqref{collision:eq:denominator} give the inequality.
\item Independence of the bits lets us iterate the bound. At the endpoints,
some input coordinates are uniform, and linearity makes the conditional
distribution uniform as well.
\item The difference of two independent copies has an exact representation
as a mixture of uniform distributions using the same $\theta_j$.
\item Finally, both the rank of $V$ and the dimension of changes in $T$
that preserve $V$ increase with $S$. This compares the two averages in
the required direction.
\end{itemize}

Equation~\eqref{collision:eq:goal} is also the comparison between orders one
and two in Ref.~\cite[Theorem~1]{Vijay}. The use of random coordinate subsets
and dimensions shares the structure of Step~I of that reference.
Here we pass through the averaged conditional collision probability
$\mathcal C$ and derive the needed bounds by a one-bit calculation with
fractions and induction over finitely many coordinates. Thus the comparison
required for Eq.~\eqref{eq:renyi} can be verified without relying on an
external channel-comparison theorem.

\section{An elementary proof of the coset inequality}\label{app:coset-proof}
We give a self-contained counting proof of the coset
inequality~\eqref{eq:cr}. Throughout this appendix, $G$ is a binary
linear subspace and $0<q\le1$. The proof has two steps: first, a linear
space has at least as much weight as any of its translates; second, we
apply this observation after grouping pairs of words by their sum.

\paragraph{The zero coset has maximum weight.}
For every linear subspace $H\subseteq\F^m$, every $w\in\F^m$, and every $0<s\le1$,
we claim that
\begin{equation}
 Z_{w+H}(s)\le Z_H(s).
 \label{eq:coset-zero-max}
\end{equation}
Choose a random subset $S\subseteq\{1,\ldots,m\}$ by including each
coordinate independently with probability $s$, and define the coordinate
subspace
\begin{equation}
 V_S:=\{x\in\F^m:x_i=0\text{ for all }i\notin S\}.
 \label{eq:coset-coordinate-space}
\end{equation}
This random subset is only an auxiliary device for proving the
inequality; it is not a random choice made by the protocol.
A fixed vector $x$ belongs to $V_S$ precisely when all its nonzero
coordinates are included in $S$. Hence
\begin{equation}
 \mathbb E\bigl[|(w+H)\cap V_S|\bigr]
 =\sum_{x\in w+H}\Pr[x\in V_S]
 =\sum_{x\in w+H}s^{\wt(x)}
 =Z_{w+H}(s).
 \label{eq:coset-random-count}
\end{equation}
For each fixed $S$, the intersection $(w+H)\cap V_S$ is either empty
or a translate of $H\cap V_S$. Indeed, if $x_0$ belongs to that
intersection, then the intersection is exactly $x_0+(H\cap V_S)$.
Its cardinality is therefore at most $|H\cap V_S|$. Taking expectations
and applying Eq.~\eqref{eq:coset-random-count} also with $w=0$ proves
Eq.~\eqref{eq:coset-zero-max}.

\paragraph{Pairs with a fixed sum.}
Now take $G\subseteq\F^n$ and fix any $z\in\F^n$. Let
\begin{equation}
 I_z:=\{i:z_i=0\},\qquad
 \pi_z:\F^n\longrightarrow\F^{I_z},\qquad
 H_z:=\pi_zG,\qquad
 m_z:=|G\cap\ker\pi_z|,
 \label{eq:coset-projection}
\end{equation}
where $\pi_z$ is coordinate projection. For a pair $(a,a+z)$, a
coordinate with $z_i=1$ contributes exactly one to the sum of the two
Hamming weights. A coordinate with $z_i=0$ has equal entries in the
two words and contributes either zero or two. Thus
\begin{equation}
 \wt(a)+\wt(a+z)=\wt(z)+2\wt(\pi_za).
 \label{eq:coset-pair-weight}
\end{equation}
This explains why the weight parameter becomes $q^2$ on the projected
coordinates.

The image of $u+G$ under $\pi_z$ is the coset $\pi_zu+H_z$.
Every word in this image has exactly $m_z$ preimages in $u+G$, since
each nonempty fiber is a translate of $G\cap\ker\pi_z$.
Consequently, Eq.~\eqref{eq:coset-zero-max} with $s=q^2$ gives
\begin{align}
 \sum_{a\in u+G}q^{\wt(a)+\wt(a+z)}
 &=q^{\wt(z)}m_zZ_{\pi_zu+H_z}(q^2)\notag\\
 &\le q^{\wt(z)}m_zZ_{H_z}(q^2)\notag\\
 &=\sum_{a\in G}q^{\wt(a)+\wt(a+z)}.
 \label{eq:coset-fixed-sum}
\end{align}
The comparison holds for every $z$, with no assumption that $z$ lies
in $G$.

\paragraph{Summing over the possible sums.}
For $a\in u+G$ and $b\in v+G$, their sum $z=a+b$ belongs to
$u+v+G$. Conversely, every $a\in u+G$ and $z\in u+v+G$ determine
a unique $b=a+z\in v+G$. Grouping the pairs by $z$ and using
Eq.~\eqref{eq:coset-fixed-sum}, we obtain
\begin{align}
 Z_{u+G}(q)Z_{v+G}(q)
 &=\sum_{z\in u+v+G}\ \sum_{a\in u+G}
       q^{\wt(a)+\wt(a+z)}\notag\\
 &\le\sum_{z\in u+v+G}\ \sum_{a\in G}
       q^{\wt(a)+\wt(a+z)}\notag\\
 &=Z_G(q)Z_{u+v+G}(q).
 \label{eq:coset-pair-sum}
\end{align}
The last equality uses the same bijection with $a\in G$ and
$b\in u+v+G$. This proves Eq.~\eqref{eq:cr}.

All counting statements remain valid for zero-dimensional spaces or
an empty coordinate set $I_z$; in the latter case $H_z=\{0\}$ and
$m_z=|G|$. The endpoint $q=1$ is included directly: every coset has
weight $|G|$, and Eq.~\eqref{eq:cr} is an equality. The two further
comparisons in Lemma~\ref{lem:cosets} then follow by the summations
already given in its proof.

\end{document}